\documentclass[conference]{IEEEtran}
\usepackage{cite}
\usepackage{amsmath,amssymb,amsfonts,amsthm}
\usepackage{algorithmic}
\usepackage{graphicx}
\usepackage{xcolor}
 \usepackage[english]{babel}
\usepackage[utf8]{inputenc}
 \usepackage{mathtools}
\usepackage{float}
\usepackage{enumitem}
\usepackage{tikz}

    \usepackage{hyperref}

\usepackage[ruled,vlined]{algorithm2e}

\DeclareMathOperator*{\cost}{cost}
\DeclareMathOperator*{\val}{val}
\DeclareMathOperator*{\brb}{brb}
\DeclareMathOperator*{\fee}{fee}

\newtheorem{definition}{Definition}

\newtheorem{prop}{Proposition}

\begin{document}

\title{Hashcashed Reputation\\  
\huge{with Application in Designing Watchtowers}}

\author{\IEEEauthorblockN{Sonbol Rahimpour}
\IEEEauthorblockA{\textit{Electrical and Computer Engineering} \\
\textit{University of Alberta}\\
Edmonton, Canada \\
rahimpou@ualberta.ca}
\and
\IEEEauthorblockN{Majid Khabbazian}
\IEEEauthorblockA{\textit{Electrical and Computer Engineering} \\
\textit{University of Alberta}\\
Edmonton, Canada \\
mkhabbazian@ualberta.ca}
}

\maketitle

\begin{abstract}
  We propose a novel reputation system to stimulate
  well-behaviour, and competition in online markets.
  Our reputation system is suited for markets where a publicly-verifiable ``proof-of-misbehaviour'' can be generated when one party misbehaves.
  Such markets include those that provide blockchain services, such as monitoring services by watchtowers.
  Watchtowers are entities that watch the blockchain on behalf of their offline clients to protect the clients' interests in applications such as payment networks (e.g., the Lightning network).
  In practice, there is no trust between clients and watchtowers, 
  and it is challenging to incentivize watchtowers to well-behave (e.g., to refuse bribery). 
  To showcase our reputation system, in this work, we create an open market of watchtowers, where watchtowers are motivated to not only deliver their promised service but also reduce their service fees in competition with each other. 
\end{abstract}


\section{Introduction}
Hashcash is a cryptographic hash-based proof-of-work algorithm proposed by Adam Back~\cite{back1997partial} to limit email spamming.
Over time hashcash has found other applications in mitigating denial of service attacks~\cite{mankins2001mitigating,dwork1992pricing}.
Today, it is perhaps most known for its role in the consensus algorithms of Bitcoin and other cryptocurrencies.
In this work, we introduce yet another application of hashcash: reputation system.

Reputation systems play a profound role in online markets, where members have no prior real-world interactions with each other. 
An effective reputation system enhances the honesty of the members, 
and incentivizes them to follow higher standards to maximize their profit. 
Nowadays, there are many different reputation systems, each used for a specific target market. For example, a common reputation system is the recommendation system, used by online shopping stores such as Amazon and eBay. 
In this system, buyers rate products or services they receive; in some cases, sellers can rate buyers, too.
For example, on eBay, each party can assign a positive or negative rate to its counter-party.
The cumulative rating of each party is then used to create a public reputation for a member or a product.

In this work, we introduce a new reputation system, which targets markets that provide blockchain services such as monitoring services.
These markets can often be designed in such a way that a party can generate a \emph{publicly verifiable} ``proof-of-breach'' if 
the service provider does not fulfill their contractual obligations.
As an application, we use our reputation system to design an open watchtower market for the Bitcoin Lightning network.

Watchtower is a service to protect users who participate in payment channels~\cite{watchtowerLN}.
Payment channels, such as Lightning channels~\cite{poon2016Bitcoin}, enable parties to perform so-called off-chain transactions (outside the blockchain) through private communications.
Payment channels guarantee the security of off-chain transactions using allocated collateral.
To maintain the security of the payment channel, however, 
 parties need to be frequently online and watch the blockchain for possible frauds by the counter-parties;
a party that goes offline risks losing payments as the counter-party can close the channel using an outdated channel state.
A party who decides to go offline can employ a third-party, referred to as watchtower, 
to watch the blockchain and protect the channel on the party's behalf.

Incentivizing watchtowers is a non-trivial task. 
One approach to incentivize watchtowers is to pay them at the beginning of the watching service we expect them to provide in the future.
The challenge with this approach is to make sure that watchtowers will, in fact, fulfill the service for which they were paid.
There are known solutions~\cite{pisa,sprites} that can handle this using complex smart contracts. 
Such solutions are, however, not applicable to Bitcoin’s simple script language.
In addition, they generate extra load on the blockchain, which is what payment channels try to avoid in the first place.
Another approach is to pay watchtowers a fee only when they detect a fraud~\cite{avarikioti2018towards,unlinkdryja,watchtowerLN}.
A major issue with this approach is that watchtowers will not receive any fee if (thanks to them) no fraud occurs.
In fact, this approach may incentivize watchtowers to encourage frauds. 

In this work, we show how our reputation system can be used to incentivize watchtowers.
At a high level, our solution works as follows.
Consider an open market, where any watchtower can join to provide watching services for profit.
Each watchtower has a reputation, which is basically a proof-of-work tied to its ID. The ID of a watchtower is essentially a public key selected by the watchtower itself.
A client (i.e., a payment channel user) first selects a watchtower. 
The selection criteria are designed in such a way that a watchtower with a higher reputation gets a higher chance of being selected.
This incentivizes watchtowers to progressively improve their proof-of-work to gain a higher share of the market.
After selecting a watchtower, the client communicates the terms of service with the watchtower.
Accordingly, the watchtower creates a contract\footnote{This is not a smart contract. It is rather a stand-alone digital document that is exchanged between a client and a watchtower through private communications.} (a stand-alone digital document) and signs it.
The contract is designed in such a way that 1) it facilitates its exchange with money, 2) it can be used by the client to generate a publicly verifiable
proof-of-breach if the watchtower does not fulfill its terms. 
Such a contract incentivizes a watchtower to perform its service as otherwise, it can lose its reputation, hence its share of the market.
Of course, a defaulting watchtower can start over by creating a new ID and a new proof-of-work reputation.
However, this will be very costly in a competitive market,
where watchtowers progressively strengthen their proof-of-work to increase or merely maintain their share of the market.

\section{Background}
\subsection{Hashcash}
Hashcash is a type of proof-of-work that is based on cryptographic hash functions such as SHA-256.
%In this work, we view hashcash as follows.
Let $H$ be such a cryptographic hash function.
For a given string $s$, we say $h$ is an $n$-bit hashcash of $s$, if $H(s || h)$ in binary has $n$ leading zeros, 
where $s || h$ indicates the concatenation of $s$ and $h$.
One can easily verify a hashcash $h$ by computing 
$H(s || h)$, and then counting the number of leading zeros of the result.
To find an $n$-bit hashcash, however, one requires to compute the hash function $2^n$ times, on average. 
%It is because the probability that $H(s||h')$ has $n$ leading zeros for a random number $h'$ is $\frac{1}{2^n}$.
Therefore, finding an $(n+1)$-bit hashcash requires, on average, twice as much work 
as finding an $n$-bit hashcash.
The value of $n$ is used as the measure of the amount of work performed.

\subsection{Payment channels}   
    Bitcoin can process a small number (about 7) of transactions per second \cite{onscaling}.
    In addition, it exhibits high transaction latency; a transaction needs at least ten minutes, on average, to go through.
    These are major impediments to further adoption of Bitcoin. 
    
    Payment channels~\cite{decker2015fast,poon2016Bitcoin} are a promising solution to the low throughput and high latency of Bitcoin. 
    Payment channels achieve this by handling most transactions outside the Bitcoin blockchain.
    A payment channel can be viewed as a temporary joint account between two parties, say Alice and Bob.  
    The channel is opened by a Bitcoin transaction, referred to as the opening transaction ($topen$).
    The opening transaction commits the parties UTXOs (Unspent Transaction Outputs) into a single 2-of-2 multisig output,
    controlled jointly by Alice and Bob.
    Once the channel is opened, Alice and Bob can exchange private transactions off the chain.
    These off-chain transactions, called commitment transactions ($ctx$), commit the output of $topen$ into a set
    of outputs that divide the channel fund between Alice and Bob.
    
    For example, suppose Alice and Bob each deposit two Bitcoins to open a payment channel.
    The first commitment transaction, $ctx_1$, divides the total channel fund of four Bitcoins evenly between Alice and Bob.
    This commitment transaction ensures that each party can get their money back in case the counter-party disappears
    after the channel is opened. 
    Now, suppose that Bob wishes to purchase a good worth of one Bitcoin from Alice.
    To make the payment, Bob provides Alice with a new commitment transaction, $ctx_2$, which 
    updates the channel balance by giving Alice three Bitcoins and giving Bob one Bitcoin.
    In addition to this, Bob must revoke $ctx_1$, as this transaction now reflects an outdated balance. 
    To this end, Bob will give Alice a so called justice transaction $jtx_1$.
    The justice transaction is used by Alice to penalize Bob if he publishes $ctx_1$ in the blockchain.
    To enable this punishment mechanism, 
    commitment transaction is designed such that once published by one party,
    they give the counter-party a dispute period during which the counter-party can send a justice transaction if there is any.
    In our example, if Bob publishes $ctx_1$,
    Alice can dispute it using $jtx_1$ and collect the whole channel fund.

 \subsection{Watchtowers}
  \label{sec:watchtower}
   The punishment mechanism explained earlier requires each party to stay online and monitor the blockchain for
   possible cheating by the counter-party.
   Alternatively, a party may delegate the task of monitoring the blockchain to a third party called watchtower.
   In practice, this is accomplished by giving the watchtower the first 16 bytes of every ctx's transaction ID ($ctx_{txid}$),
   as well as every justice transaction encrypted using the second 16 bytes of the corresponding $ctx_{txid}$.
   If the watchtower finds a transaction on the blockchain with an ID whose 16-byte prefix matches a prefix, it has stored,
   it will decrypt the corresponding $jtx$ transaction using the second 16 bytes of the transaction ID, 
   and then broadcasts $jtx$ to the network to penalize the cheating party.
   Note that a watchtower cannot identify a channel in this design unless one of the two channel's owners cheats.

   There is a special case where one of the two parties does not need to watch the blockchain, hence does not require a watchtower when the party goes offline.
   This case is when Alice opens a payment channel with Bob and uses this channel only to pay Bob.
   In other words, Alice does not receive/accept any payment from Bob on the channel.
   We call such a channel a \emph{directional payment channel} from Alice to Bob.
   For a directional payment channel from Alice to Bob, Alice does not need to watch the blockchain for old commitment transactions.
   It is because Bob does not have any incentive to cheat, as old transactions give Bob less money than the latest commitment transaction.
   Note that unlike Alice, Bob needs to watch the blockchain as Alice has an incentive to cheat by claiming an old commitment transaction.

\subsection{The Lightning network}     
  
  Two parties may not have a payment channel between themselves, but they may be connected through multiple payment channels. 
  For example, Alice may not have a payment channel with Bob, but she may have a payment channel with Charlie, who has a payment channel with Bob.
  In this case, the Lightning network enables Alice to transfer money to Bob through Charlie.
  
  The challenge to transfer money from Alice to Bob through Charlie is that the transfer from Alice to Charlie and the one from Charlie to Bob are independent.
  Consequently, if one of these two transfers goes through, there is no guarantee 
  that the other one will go through.
  The Lightning network handles this issue by binding the two transfers using a Hashed TimeLock Contract (HTLC).
  Using HTLC, the two transfers on the way from Alice to Bob are conditioned on Bob releasing a secrete preimage.
  This essentially ensures that Bob's secret preimage and Alice's money are atomically exchanged.

  %HTLC’s are a type of smart contract that use preimage resistance of %cryptographic hash functions, along with timelocks,

 \section{The proposed reputation system}
  \label{sec:rep}  
  
\subsection{System components}  
\label{sec:sysComp}

  \textbf{Market.} 
  A market is identified by a number called the \emph{market ID}, and is composed of servers, which are entities that provide the service to clients for profit.
  The markets we consider are open, which means they allow any server to join and offer their service.
  Every server is identified by an ID, which is a public key. 
  The market has no central authority to assign IDs to servers. Therefore, similar to Bitcoin users, servers select their IDs on their own.
  As a result, an entity can enter the market with many different IDs, as it can create many public keys.

  \textbf{Reputation.} In our system, a server generates its own reputation and proves it using a hashcash.
    More specifically, given a server ID, a market ID, and a hashcash the reputation of a server
    is calculated as the number of leading zeros of 
    \[
    H(server ID || market ID || hashcash)
    \]
    in binary, where $H$ is a cryptographic hash function (such as SHA-256).
    This is illustrated in~Fig.~\ref{fig:rep}.
    Note that the hashcash is basically a nonce, similar to the nonce in the Bitcoin block headers.  
    Also, note that reputation is tightly linked to a single pair of server ID and market ID.
    This prevents an entity from using a reputation for multiple IDs or over multiple markets.
    
    A server can increase its reputation on its own by creating a better hashcash. 
    This is in contrast to the existing reputation systems, where a server's reputation is increased when it receives good reviews from clients/customers. 
    Finding a better hashcash, however, is not free. The server has to mine itself, or rent mining power (whichever the server finds more cost-effective). 
    A server may make such an investment to, for example, get a better share of the market or merely maintain its share in a competitive (but profitable) market. 
    We remark that anyone can join the market but may not necessarily profit from the market because of the competition that exists between servers to provide the service at the lowest possible fee.
    It is not our intention to
    guarantee
   that the market will have many servers in it. In fact, the market may become dominated by a small number of  ``powerful'' servers. 
    However, we aim at designing a market where every server has a
    strong incentive to fulfill its contracts, and lower its service fee in competing with other servers.
    %At a high level, we achieve these by 1) making the market open, which means that any server can join the market at anytime, 2) diminishing a server's reputation if it breaches a contract, and 3) allowing clients to select servers b their reputation as well as their service fee.

    \begin{definition}[Reputation cost]
    \label{def:repCost}
       The reputation cost, $\cost(r)$, is an estimate of the minimum energy (electricity) cost to generate a server ID with reputation~$r$. 
      The cost of a server~$s$ with reputation~$r$ is denoted $\cost(s)$,
      and is defined to be equal to $\cost(r)$.
      We remark that reputation cost is time variant, because energy cost and hardware efficiency change over time. 
    \end{definition}

    \begin{figure}
    \centering
        \includegraphics[width=0.45\textwidth]{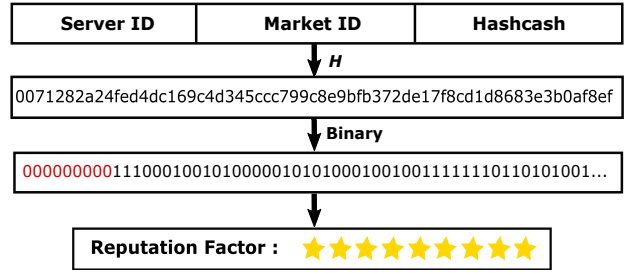}
        \caption{ An illustration of how a reputation is calculated.}
    \label{fig:rep}
    \end{figure}

    \textbf{Documents.} The system generates two types of digital documents: contracts, and proof-of-breaches. 
    A contract is a stand-alone digital document that includes a server ID, terms of the service, and the server's signature.
    A proof-of-breach, on the other hand, is a digital document that consists of a contract and publicly-verifiable evidence proving that the server who signed the contract has breached it.

    Each contract contains two hash images: 
    a \emph{client hash image} selected by the client and a \emph{server hash image} selected by the server.
    A contract is considered valid only if it is presented along with the preimage of the server's hash image.
    In contrast, a proof-of-breach becomes invalid if the preimage of the client's hash image is presented.
    As will be explained later, this mechanism allows a digital document to be atomically exchanged for cryptocurrency.   
    % In the case of proof-of-breach, the client's preimage is used by a server to invalidate a proof-of-breach if, for example, 
    % the server has come into a settlement with the client.
    One can view these preimages as one-time use on/off switches:
    the server's preimage activates the contract, while the client's preimage terminates it.

    \textbf{Distributed storage system.} The reputation system utilizes a distributed storage system to store servers' records, including IDs, hashcash, and proof-of-breaches.
    The main requirement of this storage system is to ensure that each server's record is stored by \emph{at least one node}, who is willing to share the record with others.
    % This is a relaxed condition compared to other distributed system (e.g., Bitcoin) which require majority of nodes to be honest.
    As stated in the next proposition, this is a condition that is naturally achievable in our system.
    Consequently, the proposed reputation system can rely on the set of servers/clients as part of a (perhaps larger) distributed storage system 
    that stores servers' records.    
    \begin{prop}
      For every ID, hashcash, and proof-of-breach, there is at least one node in the system that has incentive to store and share the record.
    \end{prop}
    \begin{proof}
        A server has full incentive to store and share its own ID and hashcash with others.   
        With regards to a proof-of-breach, there is at least one node with a strong incentive to store and distribute the record: 
        the victim of the contract breach.
        In addition, servers have incentives to store a proof-of-breach against the defaulting server. It is because the market share of the defaulting server becomes available to all others when the defaulting server is out of the market. Note that every server in this market essentially competes with every other server, because all servers offer their service to the same pool of clients. 
    \end{proof}

    % The reputation system requires a server to store its ID, hashcash, as well as any preimage that invalidates proof-of-breaches made against the server.

\subsection{Interactions}
  %The three main components of the proposed reputation system are clients, servers, and a distributed storage system.
   The proposed reputation system supports two types of inter-component interactions: client-server interactions and client-storage interactions.   
   Client-server interactions are to transfer reputation information including server ID, 
   hashcash and server preimages (i.e., preimages that can invalidate existing proof-of-breaches).
   They are also to communicate terms of services and fees and to transfer contracts.
   Client-storage interactions, on the other hand, are to store and retrieve proof-of-breaches from the distributed storage system.

\subsection{Protocol}
   Consider a market with a set of servers that provide a service for profit. 
   Each server has an ID and a hashcash to represent its reputation.
   In this market, a client who wishes to receive the service goes through the following steps.
   \begin{enumerate}
     \item \textbf{Screening}: The client collects reputation information (including server IDs and hashcash) from the servers.
       In addition, it retrieves proof-of-breaches from the distributed storage system and verifies them.
       Valid proof-of-breaches can be cached at the client side.
       The client discards servers with a valid proof-of-breach against them. 
       The servers that are not discarded are referred to as \emph{candidate servers}.
      \item \textbf{Negotiation}: The client negotiates the terms of service and fees with the candidate servers through private communications.
      Alternatively, instead of actively engaging with each client, servers may offer fixed service plans that are ready to be signed.

      \item \textbf{Selection}: Considering the servers' reputations and their service fees, 
        a client selects a subset of the candidate servers to contract with.
        If the client selects more than one server, it will contract with each candidate server separately.
        We assume that the client receives no damage if at least one of the 
        selected servers fulfills the terms of the contract.
        For example, in a watchtower market, if at least one watchtower respects its  contract and monitors the blockchain, the payment channel is fully protected. 
        If all the selected servers deny the fulfillment of the contract, however, the client can create a proof-of-breach against every single selected server.

        In this work, we do not impose any specific method for selecting the candidate servers.
        In fact, it may not be possible to enforce  a fixed selection method, as clients may have other reasons (external to the system) to select a particular server.
        Nevertheless, we suggest the following  properties to be considered in designing any selection method.
        \begin{enumerate}
            \item \textbf{Reputation-aware}: if the algorithm selects a server with reputation $r$ and service fee $f$, then it must also select any server with reputation $r'>r$ and service fee $f'\leq f$;
            \item \textbf{Fee-aware}: if the algorithm selects a server with reputation $r$ and service fee $f$, then it must also select any server with service fee $f'<f$, and reputation $r'\geq r$.
            \item \textbf{Damage-aware}: for every selected server $s$, the contract's value 
            for the client, $\val(c)$,  must be less than $\frac{\cost(s)}{k}$, 
            where $k\geq 1$ is a security parameter defined in Section~\ref{sec:Adv}. 
        \end{enumerate}
        The first two properties stimulate competition, 
        and encourage servers to increase their reputation and reduce their service fees. Note that both reputation and service fee are considered by a client in selecting servers. 
        The third property encourages servers to behave, and is a defence mechanism against bribery, as will be explained in Section~\ref{sec:secAnal}.         
        
      \item \textbf{Purchase}:  To purchase a contract, the client first obtains a signed copy of the contract from the server and verifies the contract.
      Then, the client purchases the contract by, for example, atomically exchanging the server's preimage for cryptocurrency. 
      Recall that a signed contract is considered valid only when it is  presented with this preimage.
      \color{black}

      \item \textbf{Punishment}: If the client ever discovers a breach in the contract, 
        it creates a proof-of-breach and stores it in the distributed storage system.
               Optionally, the client can negotiate terms of settlements with the defaulting server. 
        The contract supports the atomic exchange of the client's preimage (which invalidates the proof-of-breach) for cryptocurrency.
        In our model, a proof-of-breach against a server will reduce the server's reputation to zero, unless the server provides the corresponding client's preimage indicating that the contract has been terminated as a result of, for example, a settlement. 
   \end{enumerate}

   The main challenge in enabling the proposed reputation system is to enable publicly verifiable proof-of breaches.
   We believe that many markets that provide blockchain services can be (re)designed to provide this feature.
   In this work, we present one such service: blockchain monitoring.

\section{Adversarial model}
\label{sec:Adv}
  An adversary can join the market with an arbitrary number of IDs.
  Moreover, a client cannot determine if a set of IDs belong to the same entity. 
  We assume that there are at least two independent service providers (servers) in the market\footnote{
  In an open market, if the market is profitable for a single provider, it makes sense for another service provider to join the market.}.

  An adversary can launch a denial of service attack by populating a storing node with server IDs and/or proof-of-breaches.
  It may also bribe a storing node to delete its proof-of-breach from its storage. 
  %However, we assume that if an honest storing node has a record in its storage,
  %the adversary cannot prevent the record from being publicly accessible.
  An adversary can bribe a server to breach a contract.
  However, we assume that at each point in time, a server has standing bribes on at most $k$ different contracts, where $k\geq 1$ is a system security parameter.
  In addition, we assume that the total amount of bribe offered to a server to breach a contract $c$ is not more than the value of $c$ for its client
  (otherwise, the bribe can be used to buy off the client directly!).
  Moreover, we assume that a server $s$ does not accept 
  any of its standing bribes, if the total amount of bribe offered to $s$ (over all the contracts $s$ is handling) is
  less than $\cost(s)$, as defined in Definition~\ref{def:repCost}.

  For a single contract, a client may select and pay multiple servers.
  We assume that the client does not receive any damage if at least one of these servers fulfills the terms of the contract.
  For example, in the case of the watchtower market, the client receives no damage if at least one of the paid watchtowers monitors the blockchain and follows the terms of the contract (e.g., submit the justice transaction in case of cheating).
  In selecting servers, we assume that a client has a good estimate of the reputation cost function.
  Finally, we assume that the client can create a proof-of-breach against all paid servers if there is a term in the contract that
  is not fulfilled by any of the paid servers.

\section{Security Analysis}
\label{sec:secAnal}
  \textbf{Sybil attack.} In the Sybil attack, an adversary attempts to subvert the reputation system by creating multiple IDs.
  In our reputation system, an adversary cannot impact the reputation of servers by merely creating many IDs.
  It is because, unlike other reputation systems such as recommendation systems,  in our system, the reputation of a server is not effected by other nodes.
  In fact, there are only two things that can impact one's reputation: proof-of-work and proof-of-breach.

  \textbf{DoS attack.} An adversary may create multiple IDs and/or fake contracts to use up the storage of storing nodes.
  For instance, suppose that there are nodes in the system that provide clients with server IDs, hashcash, and their IP addresses. 
  An adversary may try to overwhelm these nodes by flooding them with server IDs.
  At some point, a node does not have enough storage room to accept new IDs or has to replace old IDs with new ones coming.

  A simple counter-measure against this type of denial-of-service attacks is to use servers' reputations to prioritize records.   
  For example, a proof-of-breach against a highly-reputable server has a higher priority than a proof-of-breach against a normal server.
  With such prioritising, a storing node accepts and stores a new record if either the node has enough room or the priority of the new record is higher than the priority of the lowest-priority record in the storage.
  In the latter case, the lowest-priority record is replaced with the new record.
  
  By the following proposition and considering today's computational power, storage capacity, and electricity cost it is impractical for an adversary to flush out the records of all honest servers from a storage node.

  \begin{prop}
    An adversary requires on average $M\cdot 2^{r_{max}}$ hash computations to flush out the records of all honest servers from a storage node,
    where $M$ is the number of records the node can store, and $r_{max}$ denotes the maximum reputation of any honest server.
  \end{prop}
  \begin{proof}
    For an adversary to flush out the records of all honest servers from the storage, it needs to create $M$ records, each with reputation of at least $r_{max}$.
    To create a reputation of at least $r_{max}$, the adversary needs to compute on average $2^{r_{max}}$ hashes.
    Therefore, in total, the adversary needs to compute at least $M\cdot 2^{r_{max}}$ hashes, on overage.
  \end{proof}

  For example, assuming that $r_{max}>36$, and\footnote{Even a personal computer, which is highly inefficient in mining, can generate an ID with reputation of at least $36$ in few hours.} the storage node can store $2^{40}$ records, the adversary must compute at least $2^{76}$ hashes in order to remove/replace all the honest servers' records. Using an ASIC hardware (with the speed of 10 tera hashes per second), this takes about 170 years!

  \begin{prop}
    A defaulting server has no incentive to flush a proof-of-breach against itself out of a storing node by generating a new ID and many artificially-generated proof-of-breaches against the new ID. 
  \end{prop}
  \begin{proof}
    Suppose that the new artificially generated proof-of-breaches pushes out the valid proof-of-breach out of a storage node.
    Since the storage node prioritizes records according to the reputation of the corresponding server ID, the reputation of the new ID must be at least equal to the reputation of the defaulting server.
    There is no incentive then for the defaulting server to try to push out the proof-of-breach rather than starting fresh with the new ID.
  \end{proof}

  A storage node needs
  to verify records such as proof-of-breaches.
  Such verification needs little but non-negligible resources such as computation and memory.
  An attacker can overwhelm a storage node by sending many fake proof-of-breaches to the node.
  A storage node can mitigate this type of DoS attacks by simply requiring a small proof-of-work along with a submitted record.  
  This mitigation is, in fact, the original application of the hashcash, which is to  mitigate such DoS attacks.

\textbf{Fake resolved proof-of-breaches.} 
A server may create valid proof-of-breaches against itself, and then resolve them to create a market image of settling all disputes with clients.
This method, however, does not impact the server's reputation because terminated/resolved contracts neither increase nor decrease the server's reputation.

  \textbf{Bribery.} A defaulting server may attempt to bribe the storing nodes to delete a proof-of-breach.
  Bribing, however, cannot guarantee that a digital document is purged.
  In fact, the victim of a contract breach would always keep its proof-of-breach and can re-distribute it at any time. 
  The main hope of a defaulting server who is willing to remain in the market is to settle with the client (by purchasing the client's preimage, which invalidates the contract) or start over with a new reputation.
  
  Another type of bribery is when an adversary offers a bribe to a server to breach a contract.
  For instance, a payment channel party may bribe a watchtower to stop monitoring the channel for its counter-party.
  This type of bribery is a serious threat against any digital market.
  Note that since a server provides services to many clients, it may receive multiple bribes from multiple adversaries. 
  Nevertheless, Proposition~\ref{prp:brife-safe} shows that every contract in the market is ``bribe-safe'', as defined below.
  \begin{definition}
    (Bribe-safe) A contract is bribe-safe if the cost of every server responsible to execute the contract is more than the total amount of bribes offered to the server. 
  \end{definition}
    Note that the condition in the above definition is somewhat strong, because a contract is still safe even when a single paid server (as opposed to all) fulfills its contractual obligation.

  \begin{prop}
  \label{prp:brife-safe}
    Under the assumptions described in the adversarial model, every contract in the proposed system is bribe-safe if the selection methods used by the clients are \emph{damage-aware}.
  \end{prop}
  \begin{proof}
    Let $s$ be any server in the market, and $\mathcal{C}$ be the set of contracts for which the server has a standing bribe offer.
    Let $\brb(c)$ denote the total amount of bribe offered to breach a contract $c$, $\val(c)$ denote the maximum value of $c$ for the client, and $\cost(s)$ denote the cost of server $s$ as defined 
    in Definition~\ref{def:repCost}.
    By the selection method's third property, we have
    \[
      \forall c\in\mathcal{C}:\quad \val(c)< \frac{\cost(s)}{k}.
    \]    
    In addition, by our adversarial model assumptions, we have
    \[
      \forall c\in\mathcal{C}:\quad \brb(c)<\val(c),
    \]       
    and
    \[
      |\mathcal{C}|\leq k,
    \]
    where $|\mathcal{C}|$ denotes the cardinality of $\mathcal{C}$.
    Therefore, we get
    \[
    \begin{split}
      B
      &\leq \sum_{c\in\mathcal{C}} \brb(c)
      < \sum_{c\in\mathcal{C}} \val(c)\\
      &\leq \sum_{c\in\mathcal{C}} \frac{\cost(s)}{k}
      =\frac{|\mathcal{C}|}{k}\cost(s)
      \leq \cost(s),
    \end{split}
    \]
    where $B$ is the total amount of bribe offered to $s$.
    Thus, the total amount of bribe offered to $s$ is less than
    $\cost(s)$.
    Therefore, by our adversarial model's assumption, the server does not accept a bribe to breach any contract $c\in\mathcal{C}$.
  \end{proof}

%  For instance, a storage node that stores server IDs and/or proof-of-breaches prioritizes them according to the server's reputation.
%  In particular, if a storage node does not have enough space to store all the records, it would store those with higher server's reputation.

\section{A reputation-based market of watchtowers}
%  In this section, we discuss how our reputation system can be employed to incentivize watchtowers to watch the blockchain on behalf of clients.
%  To this end, we need to show how contracts and proof-of-breaches are generated.
%  

    Our reputation system can be used in various digital markets which offer blockchain services.
    In this section, we show one example where we apply our system in a market of watchtowers that provides blockchain monitoring service to payment channel holders who decide to go offline.

    Consider a market of watchtower servers, where each server has an ID (public key), and a reputation as defined in Section~\ref{sec:sysComp}.
    The market utilizes a distributed storage system, which stores the server's records, including ID, hashcash, IP address, proof-of-breaches, and preimages.
    Recall that servers and victims of contract breaches participate in this distributed storage system, and for each record, there is at least one node that has a strong incentive to store the record and distribute it.
    
    \textbf{Screening.} Consider a payment channel between Alice and Bob and suppose that
    Alice is interested to pay one or more watchtowers to monitor the blockchain on her behalf while she is offline.
    First, in a screening process, Alice contacts the storage system, collects all servers' records, and evaluates servers.
    A server's evaluation includes a single computation of the hash function to calculate the server's reputation, and verification of proof-of-breaches against the server if there is any.
    Verification of a proof-of-breach is expected to be harder than a single computation of the hash function.
    However, a client needs to verify a proof-of-breach at most once, as it can cache the result for later use.
    The watchtowers that successfully pass Alice's evaluation are referred to as candidate watchtowers.  
    
    \textbf{Negotiation.} In this step, Alice communicates the terms of a contract
    with the candidate servers, and negotiates for service fees.
    At the end of this step, Alice knows how much fee each 
    candidate watchtower charges for the contract.
    Fig.~\ref{fig:contract} shows a watchtower contract sample, which
    consists of a market ID, server ID, server's hashcash, a set of 16-byte transaction ID prefixes,
    a set of encoded justice transactions, the range of
    blocks that the watchtower has to monitor,
    server's hash image, client's hash image, and server's signature.
    The client's and server's hash images are hashes of random numbers generated by, respectively, the client and the server.
    Note that contract's data does not reveal any information about the client.
    
\begin{figure}[t]
    \centering
    \includegraphics[width=0.48\textwidth]{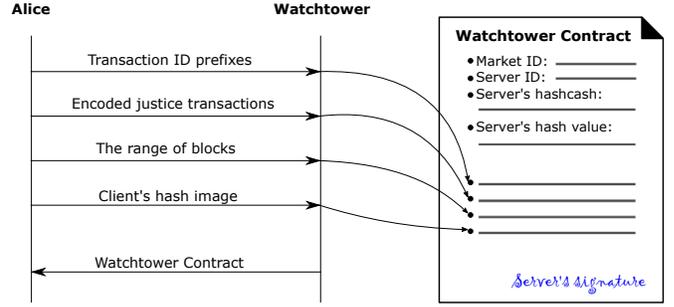}
    \caption{A sample of watchtower contract.}
    \label{fig:contract}
\end{figure}

    \textbf{Selection.} Considering the reputations of the candidate watchtowers and their service fees, 
    Alice selects a set of watchtowers to contract with. 
    As mentioned earlier, we do not enforce any specific selection algorithm.
    We, however, suggest any selection method to be 1) reputation-aware;
    2) fee-aware, and 3) damage-aware.
    Algorithm~\ref{alg:sel} shows one possible selection method.
    In this method, Alice first determines the contract's maximum value,
    $\val(c)$.
    This value should be set to at least the payment channel fund,
    which is the maximum amount Alice would lose on the channel if Bob cheats. 
    Then, among candidate watchtowers whose reputation cost is more than the threshold $T=k\cdot\val(c)$, the algorithm chooses the one with the minimum fee.
    If there are multiple such watchtowers with the minimum fee, one with the maximum reputation is selected. 
    Thus, if the algorithm selects a server $s$ with reputation $r$, and service fee $f$, then for every server $s'$ we have
    \[
      f'\geq f  \quad \text{OR} \quad r'\leq r,
    \]
    where $f'$ and $r'$ denote the service fee and reputation of $s'$, respectively.
    This implies that the algorithm is both 
    reputation-aware and fee-aware.
    In addition, the algorithm is damage-aware as the reputation cost of the selected server is at least
    $k$ times the contract's value.
    
    When selection is complete, Alice contacts 
    the selected servers and provides them with her hash images\footnote{Different hash images are used for different servers.}. 
    In response, every server returns a signed contract.
    Note that a signed contract is only considered valid when it is presented with the server's preimage.
    To finalize the contract, Alice pays each server to receive the server's preimage.

    \begin{algorithm}
    \SetAlgoLined
    \KwData{A threshold $T$, and a set of servers $S$}
    \KwResult{$s^\dag\in S$}
     initialization\;
     $s^\dag \gets $ any member of $S$ with reputation cost of at least $T$\;
    \For{ every $s\in S$} {
      \If{$(\cost(s)\geq T) \,\&\,(\fee(s)\leq \fee(s^\dag)$} {
        \If{$(\cost(s)> \cost(s^\dag)\,||\,(\fee(s)< \fee(s^\dag)$} {      
          $s^\dag \gets s$\;
        }
      }
     }
     \caption{A watchtower selection method.}
     \label{alg:sel}
    \end{algorithm}

  \textbf{Purchase.}  To purchase the server's preimage, Alice can use the Lightning network, 
    or a payment channel with the watchtower.
    This may seem paradoxical, as the payment channel that Alice uses to pay the watchtower has to be protected by a watchtower, too.    
    To get around this issue, Alice uses a directional payment channel to pay the watchtower.
    As explained in Section~\ref{sec:watchtower}, Alice does not need to monitor the blockchain for a directional payment channel,
    i.e., a channel that she uses only to make payments.
    
    To purchase the contract using a directional payment channel, 
    Alice uses the watchtower's hash image in her HTLC.
    This, as discussed earlier, ensures that the contract and the service fee are atomically exchanged.
    As soon as the payment goes through (i.e., once Alice receives the watchtower's preimage), the contract becomes valid, and Alice can go offline.
    %Note that Alice can pay multiple watchtowers to watch the blockchain on her behalf.

    \textbf{Proof-of-breach.} 
    Since the Bitcoin blockchain is public, anyone can check if, for example, the $ctx$ transaction has been published. Therefore, anyone can verify whether a given contract has been breached.
    In particular, a contract together with the server's preimage
    serve as a proof-of-breach.
    Given a contract and a preimage, one can check whether 
    \begin{itemize}
      \item the format of the proof-of-breach is valid;
      \item the contract's signature is verified using the server's ID;
      \item the hash image of the given preimage is equal to the server's hash image in the contract;
      \item $ctx$ was published in one of the blocks that the watchtower was obliged to monitor;
      \item $jtx$ is a valid transaction (e.g., it can spend from $ctx$);
      \item $jtx$ was not published within the dispute period.
    \end{itemize}
    The proof-of-breach is valid if and only if all the above conditions hold. 
    Supplementary data such as the Merkle proof of the existence of $ctx$   can be appended to a proof-of-breach to assist light clients (e.g., app-based mobile clients) with the verification process\footnote{To show that $jtx$ does not exist, the Merkle proof of the existence of a transaction (other than $jtx$) that spends from $ctx$ can be provided. This proves that $jtx$ does not exist because at most one transaction can spend the output of $ctx$.}.
    This reduces the verification process to a few signature verification (e.g., checking the server's signature), a few hash computations (e.g., to verify a Merkle proof), and some simple format checking (e.g., checking the format of the contract).

    % \textcolor{red}{How to prove it to a Lightweight node? 
    % Evidence should be attached to the document. This includes a proof that a ctx in the contract was written into the blockchain, and jtx is not there. 
    % The latter can be proved by showing that another transaction that spends from
    % ctx was written ino the blockchain.}    

    \textbf{Settlement.}
      The victim of a contract breach may settle with the defaulting server.
      This can be done readily by atomically exchanging the client's preimage for cryptocurrency.

   \section{Related work}
\subsection{Reputation systems}\label{reputation systems}
    Different reputation systems use different methods to calculate reputation. 
    We classify these methods into fact-based, and review-based methods. 
    Fact-based methods calculate the reputation of a party by solely taking the activities of the party as input.
    Review-based methods, however, calculate the reputation of a party using reviews the party receives from other parties.

 \textbf{Fact-based methods.} These methods evaluate the performance of a party using a predefined function that takes the party's activities as input~\cite{huang2019repchain,li2018crowdbc,incentivizingblockchainminers,ERC}, \cite{areputationmanagement,Repucoin,zhuang2019proof}. Therefore, in fact-based methods, the reputation of a party is only a function of its activities. For example, in~\cite{huang2019repchain} the reputation of a party is increased or decreased based on the value of transactions on which the party was honest or dishonest, respectively. 
 Another example is CrowdBC~\cite{li2018crowdbc}, a reputation system with two types of parties: requester and worker. A requester is a party that offers a task, while a worker is a party that performs tasks to improve its reputation. In CrowdBC, a requester and a worker negotiate and generate a smart contract. This contract has an evaluation-function that evaluates the performance of the worker. CrowdBC uses the output of this evaluation-function as well as the average reputation of all other workers to calculate the reputation of the worker.

    % The reputation systems in~\cite{Repucoin,incentivizingblockchainminers} define reputation scores for miners of blockchains based on their activities. Miners in Repucoin~\cite{Repucoin} are parties that try to solve a puzzle for generating a block (called keyblock). One of the miners with high reputation scores is chosen to generate a block including transactions (microblock) based on the keyblock. After generating the microblock, the reputation scores of miners are updated based on the number of keyblocks and microblocks that each miner generated.

  \textbf{Review-based methods.}
 In this class, reputation of a party depends on the reviews and scores it receives from other parties~\cite{repontheblock,schaub2016trustless,sharples2016blockchain}. 
 For example, in online shopping stores such as eBay and Amazon, buyers can review and rate sellers and items. These rates together represent the reputation of sellers and items. Another example is Kudos~\cite{sharples2016blockchain}, which is an educational reputation currency in a blockchain that records intellectual efforts, and related reputation rewards. Academic people and institutions that award certificates or verify innovations are parties of Kudos. When a person completes a certificate, their institution sends them some amount of Kudos based on their review. In this system, the amount of Kudos a party has represents the reputation of the party.

\subsection{Proof of work}\label{pow}
Dwork and Naor~\cite{dwork1992pricing} introduced the concept of proof-of-work in 1992.
Motivated to combat junk emails, they proposed to require users to compute a moderately hard function of their messages and some additional information in order to send their messages.
They called these functions \emph{pricing functions}, and introduced several of them in their work.
In 1997, Back~\cite{back1997partial} proposed Hashcash proof-of-work system to deter email spams, and denial-of-service attacks. 
The Hashcah system is used today as part of consensus protocols in many blockchains, including Bitcoin~\cite{nakamoto2008Bitcoin}.

In addition to combating denial-of-service attacks (e.g.,\cite{mankins2001mitigating}), proof-of-work has been used to mitigate Sybil attacks~\cite{back1997partial,baza2020detecting,borisov2006computational,dwork1992pricing,rowaihy2007limiting}, protect peer-to-peer resource sharing~\cite{vishnumurthy2003karma}, and reward well-behaving users~\cite{biryukov2015proof}.
For instance, in~\cite{biryukov2015proof}, 
the authors proposed a micropayment method to reward Tor relay operators. 
In the Tor network, selfish clients may utilize the shared bandwidth of Tor relays without contributing any resources to the system in return. To mitigate such selfishness, Tor clients must submit proof-of-work shares, which Tor relays can resubmit to a cryptocurrency mining pool instead of paying cash directly. By analyzing the cryptocurrencies market prices, the authors showed that their method can compensate for a significant part of the Tor relay operator's expenses.

\subsection{Payment channels and watchtower}
Payment channels were first introduced by Satoshi Nakamoto~\cite{wikipayment}.
These channels first emerged as unidirectional for one-way payments~\cite{spilman}, then transitioned into bidirectional channels to support two-way payments.
The two common bidirectional payment channels are the Lightning Network~\cite{poon2016Bitcoin} and the Raiden Network~\cite{raiden} which operate on Bitcoin and Ethereum blockchain~\cite{wood2014ethereum}, respectively. There are several implementations of the Lightning network, including C-lightning~\cite{c-lightning}, Eclair~\cite{Eclair}, and LND~\cite{LND}.

To secure payment channels, users must frequently be online and watch the blockchain to protect their funds. It is because one party may publish an old commitment transaction while the other party is offline.
Dryja~\cite{unlinkdryja} suggested that users who decide to go offline for an extended period of time delegate the task of watching the blockchain to third parties.
Hertig~\cite{watchtowerLN} called these third parties \emph{watchtowers}.
Designing a secure, efficient, and decentralized watchtower protocol is a challenging task.

McCorry \textit{et al.} proposed Pisa~\cite{pisa}, a protocol that employs third parties called custodian to protect Sprites channels~\cite{sprites}.
In Pisa, users pay their custodians every time they make a transaction on the channel.
On the other hand, 
custodians lock a collateral fund, which they lose if they misbehave.
Avarikioti \textit{et al.} proposed the DCWC protocol~\cite{avarikioti2018towards}, in which
full nodes can act as a watchtower for multiple channels. Unlike Pisa, in DCWC, a watchtower gets paid only when it catches a fraud. 
In another research work, Avarikioti \textit{et al.} presented BRICK~\cite{Brick}, a protocol that detects and prevents fraud before it appears on the blockchain. 
To this end, BRICK employs a committee of third parties called Wardens. 
Wardens confirm the validity of each state channel and make sure that only the correct state is published in the blockchain when a dispute occurs.
The authors in~\cite{khabbazian2019outpost} proposed Outpost, a lightweight structure for watchtower that encodes justice transactions within commitment transactions rather than storing them in the watchtower. This construction saves an order of magnitude in storage over existing watchtower designs.
Finally, Avarikioti \textit{et al.} extended the Lightning network and introduced Cerberus Channels~\cite{Cerberus}. They motivated watchtowers to work honestly by rewarding them for any update on Cerberus Channels and forcing them to pay a penalty for any fraud.

\section{Conclusion}
  In this paper, we proposed a proof-of-work based reputation system to incentivize well-behaviour and stimulate competition in online marketplaces.
  An advantage of our system is that it does not rely on any blockchain or smart contracts to, for example, punish misbehaviour. Instead, it stores a proof of misbehaviour as a record in a distributed storage, which is only required to store each record in at least one node.
  This is an easy requirement to achieve since for each record, there is at least one party who is strongly motivated to store and (re)distribute it. Finally, to showcase our reputation system, we designed an open market of watchtowers. Our reputation system motivates watchtowers not only to behave according to their obligation but also compete with each other by progressively improving their reputation and by reducing their service fees.
  
  \textbf{Future research.}
  Our proposed system is flexible, and can be customized to achieve a specific need. For instance, the system does not impose any specific selection method.
  One can, therefore, design a customized selection method to achieve a certain objective in the market. 
  
  Another possible future research is to relax the system reliance on the security parameter $k$.
  One approach is to ask clients to publish a digest of their ongoing committed contracts.
  This way, a new client can, for example, avoid a server whose total commitment value is greater than its reputation cost.
  Fortunately, clients have 
  incentive to ask servers to sign such digests. They also have  
  incentive to provide these digests to other servers. 
  It is because a client does not want its server to over-commit as it increases the risk of a successful bribery.
  Also, a server has an
  incentive to record the digests of other servers, as these digests can show a potential customer that other servers have already committed to many contracts (another means to attract a customer).

% \textcolor{red}{ 
% Things go check/add:
% \begin{enumerate}
%   \item \textbf{Important:} reviewers asked for more formalism and theoretical results
%   \begin{enumerate}
%       \item what kind of formalism can we find in the literature? (so we can do similarly in this work)
%       \item Can we make some theoretical results about Sybil attack. E.g., nodes do not have incentive to join a market with multiple IDs.
%   \end{enumerate}
%   \item A server can give a simple signed ``commitment ticket'' to a client. 
%   \begin{enumerate}
%       \item  the ticket should include the value of the commitment (determined by the client), and the period of the commitment
%       \item A client has incentive to share and distribute its tickets, even if id does not use the tickets in the system in its own server selection method.
%       \item Servers have incentive to store other servers' commitment tickets.
%       \item how to formalize these things?
%   \end{enumerate}  
%   \item  ``This proof-of-breach somehow resembles a revocation certificate, and its functionality, a certificate revocation list. Thus suffering the same distribution problems.''
% \end{enumerate}
% }

\pagebreak

\bibliographystyle{IEEEtran}

\bibliography{IEEEabrv, Citations}
\end{document}